\documentclass[11pt,a4paper, oneside]{article}

\usepackage{amsmath}
\usepackage{amsfonts}
\usepackage{amssymb}
\usepackage{amscd}
\usepackage{latexsym}
\usepackage{mathrsfs}
\usepackage{amsthm}

\def\C{{\mathbb C}}
\def\e{\mathrm{e}}
\def\f{f}
\def\F{{\mathcal F}}
\def\H{{\mathcal H}}
\def\Hrond{\mathscr H}

\def\R{{\mathbb R}}
\def\SS{\mathscr S}
\def\T{{\mathbb T}}
\def\U{\mathcal U}
\def\V{\mathcal V}

\def\ii{i}
\def\wind{\mathrm{wind}}

\def\0{\mathfrak r}

\def\sgn{\mathrm{sgn}}

\newcommand{\dom}{\mathop{\mathrm{dom}}}

\newcommand{\tr}{\mathop{\mathrm{tr}}}

\def\d{\mathrm{d}}
\def\det{\mathrm{det}}

\def\CD{{C\!D}}

\def\2{\mathfrak{int}}

\def\Pv{\mathrm{P.v.}}
\def\diag{\mathrm{diag}}

\def\X{\mathsf X}
\def\D{\mathsf D}

\newtheorem{theorem}{Theorem}[section]
\newtheorem{prop}[theorem]{Proposition}
\newtheorem{lemma}[theorem]{Lemma}

\theoremstyle{definition}

\newtheorem{remark}[theorem]{Remark}

\newtheorem{rem}[theorem]{\bf Remark}

\def\bpm{\begin{pmatrix}}
\def\epm{\end{pmatrix}}
\def\bsm{\begin{smallmatrix}}
\def\esm{\end{smallmatrix}}

\DeclareMathOperator{\spec}{\sigma}

\begin{document}

\title{One-dimensional Dirac operators with zero-range interactions: \\
Spectral, scattering, and topological results}

\author{Konstantin Pankrashkin$\,^1$ and Serge Richard$\,^2$}
\date{\small}
\maketitle

\begin{quote}
\begin{itemize}
\item[$^1$] Laboratoire de Math\'ematiques d'Orsay, CNRS UMR 8628, Universit\'e Paris-Sud, B\^atiment 425, 91405 Orsay Cedex, France;\\
E-mail: {\tt konstantin.pankrashkin@math.u-psud.fr}
\item[$^2$] Graduate school of mathematics, Nagoya University,
Chikusa-ku, Nagoya 464-8602, Japan; \\
E-mail: {\tt richard@math.univ-lyon1.fr}
\\[\smallskipamount]
On leave from Universit\'e de Lyon, Universit\'e Lyon I, CNRS UMR5208, Institut Camille Jordan, 43 blvd du 11 novembre 1918, 69622 Villeurbanne Cedex, France
\end{itemize}
\end{quote}

\begin{abstract}
The spectral and scattering theory for $1$-dimensional Dirac operators with mass $m$ and with zero-range interactions are fully investigated. Explicit expressions for the wave operators and for the scattering operator are provided. These new formulae take place in a representation which links, in a suitable way, the energies $-\infty$ and $+\infty$, and which emphasizes the role of $\pm m$. Finally, a topological version of Levinson's theorem is deduced, with the threshold effects at $\pm m$ automatically taken into account.
\end{abstract}

\textbf{2010 Mathematics Subject Classification:} 81U15, 35Q41.

\smallskip

\textbf{Keywords:} Dirac operators, zero-range interaction, wave operators, index theorem.

\section{Introduction}

In a series of recent works on scattering theory, it has been shown that rather explicit formulae for the wave operators do exist and that these operators share structural properties amongst several models. Such new formulae were then at the root of a topological approach of Levinson's theorem. More precisely, it has been shown that this famous theorem, which allows one to compute the number of bound states of a physical system in terms of the scattering part of that system, is in fact an index theorem. Let us stress that an index theorem automatically means a strong robustness of the mentioned link between spectral and scattering properties under perturbations. We refer to \cite{BS,IR,KPR,KR1,KR5,KR6,PR,SB,RT1,RT2,RT3} for such explicit formulae in the context of Schr\"odinger operators, Aharonov-Bohm operators or for the Friedrichs-Faddeev model, and for their applications.

However, despite the variety of models already investigated with this approach, all these models share one common property: the operators describing them have a continuous spectrum made of one single connected part. This feature, which may seem harmless, has been in fact very convenient for the construction of the $C^*$-algebraic framework surrounding the topological approach. Therefore, one of the motivations for looking at Dirac operators was that its continuous spectrum is made of two unbounded disjoint parts. In addition, even if some Levinson's type results exist for this model (see for example \cite{C89,K90,L99,MN85}), it has never been argued that this relation is of topological nature\footnote{In \cite{L07} the analogy between Levinson's theorem for Dirac operators and the Atiyah-Singer index theorem is mentioned, but nothing is deduced from this observation.}. Thus, this work is a first attempt to derive explicit formulae for the wave operators in the context of relativistic operators, and to deduce some topological consequences of such formulae.

Now, as any first attempt, the model under consideration is rather simple (see  \cite{KR1} for its counterpart in the Schr\"odinger case). In fact, we investigate $1$-dimensional massive Dirac operators with the simplest interactions, namely the so called ``zero-range'' or ``point'' interactions. More precisely, we consider the scattering theory for a pair of self-adjoint operators $(H_0,H^{\CD})$, where $H_0$ is the usual free Dirac operator with mass $m>0$ and $H^{\CD}$ is any of the (four parameters family of) self-adjoint extensions which can be constructed from the restriction of $H_0$ to functions which vanish at $0$. Note that such extensions have already been studied in several papers, see for example \cite{alonso,BMN,dabr,cmp}, but our aim and results are different.

More concretely, after reviewing some properties of free Dirac operator $H_0$, we provide in Section~\ref{sec2} a parametrization of all the mentioned self-adjoint extensions using the machinery of boundary triples, and we obtain an explicit resolvent formula in terms of the parameters $(C,D)$.
Based on these formulae, the spectral properties of the Dirac operator with a zero-range interaction are then described in Proposition \ref{propspectral}.

In Section \ref{secscat}, we develop the scattering theory for our model. We first recall the main definition of the wave operators in the time-dependent framework of scattering theory as well as in its stationary approach. The spectral representation of the free operator $H_0$ is then provided.
Based on the stationary expressions for the wave operators, some rather explicit formulae could be derived for them in the spectral representation of $H_0$, but the results would not be very convenient (some unbounded operators would still be present).

Our main surprise, and one of the asset of this work, is that the wave operators can be computed very explicitly in another unitarily equivalent representation which we have called the upside-down representation. The reason for this name comes from the fact that in this representation the thresholds values $\pm m$ are sent to $\pm \infty$ while any neighbourhood of the points $\pm \infty$ is then located near the point $0$. In this representation, which takes place in the Hilbert space $L^2(\R;\C^2)$,
we first show that the wave operators for the pair $(H_0,H^{\CD})$ exist and that the stationary approach leads to the same operators than the time dependent approach. In addition, restricting our attention to one of the wave operators only, we provide an explicit expression  for this operator in terms of a product of a continuous function of the position operator $\X$ and a continuous function of its conjugate operator $\D=-i\frac{\d}{\d x}$, see formula \eqref{egalite}.
Note that the $\X$-factor is tightly dependent of the parameters $(C,D)$, while the factor containing $\D$ does not depend on them at all.
We also note that these factors admit a suitable asymptotic behavior, which allows one to develop the algebraic framework.

In the last section of this paper, we deduce the topological consequences of the explicit formula derived in Section \ref{secscat}. In particular, we derive a topological version of Levinson's theorem which relates the number of bound states of the operator $H^{\CD}$ to the winding number of a certain function which involves the scattering operator but also other operators related to threshold effects. The main result of this section in contained in Theorem \ref{thmtop}. Let us stress that in our approach, the threshold effects are automatically taken into account, namely we don't have to calculate separately the contributions due to the possible half-bound states located at $\pm m$. Note also that the discrepancy of the contributions of the scattering operator for positive or negative energies appears naturally in our framework.

As a conclusion, let us emphasize that the spectral and the scattering theory for Dirac operators with zero-range interactions are fully developed in this work, and that new and explicit formulae for the wave operators are also provided. We expect that such formulae as well as the topological approach of Levinson's theorem will still hold for Dirac operators perturbed by more general potentials.

\noindent{\bf Acknowledgments.} The work was partially supported by ANR NOSEVOL and GDR DYNQUA.

\section{Framework and spectral results}\label{sec2}
\setcounter{equation}{0}

\subsection{Dirac operators with boundary conditions at the origin}

Let $\H$ be the Hilbert space $L^2(\R;\C^2)$ with scalar product and norm denoted by $\langle \cdot,\cdot\rangle$ and $\|\cdot\|$. Its elements are written $f=\big(\bsm f_1\\f_2\esm\big)$. For $m> 0$,
we consider the free Dirac operator $H_0$ defined by
\begin{equation}
     \label{eq-dirac}
H_0=\bpm 0 & -1\\
1 &0
\epm\, \frac{\d}{\d x}
+
\bpm
m & 0\\
0 & -m
\epm
\end{equation}
on the domain $\dom(H_0)=\H^1(\R;\C^2)$. Here, $\H^1(\R;\C^2)$ denotes the Sobolev space on $\R$ of order $1$ and with values in $\C^2$.
A standard computation shows that $H_0$ is a self-adjoint operator in $\H$, and its spectrum equals $(\infty,-m]\cup[m,+\infty)$
and is absolutely continuous. In what follows we need an explicit expression for the Green function of $H_0$, that is, for
the integral kernel $G_0$ of its resolvent. For that purpose, for $z\in \C\setminus \spec(H_0)$ we set $k=k(z)=\sqrt{z^2-m^2}$,
where the branch of the square root
is fixed by the condition $\Im \sqrt{\lambda}>0$ for $\lambda<0$. As a consequence it follows that $\Im k(z)> 0$ for all $z\notin\spec(H_0)$.
Then, let us note that for any $z\in \C$ the equality  $(H_0-z)(H_0+z)=-\Delta-z^2+m^2$ holds on $\H^2(\R;\C^2)$; here $\Delta$ denotes the usual Laplacian in $L^2(\R;\C^2)$.
Therefore, for any $z\notin\spec(H_0)$ we have
\[
(H_0-z)^{-1}=(H_0+z) \big( -\Delta  -(z^2-m^2)\big)^{-1},
\]
and we infer from this equality that
\begin{equation}\label{Green_f}
G^0(x,y;z)= \frac{\ii }{2k}\bpm z+m & -\d/\d x\\
\d/\d x & z-m
\epm
\;\!
\e^{\ii k|x-y|}
=
\bpm
\ii \frac{m+z}{k} & \sgn(x-y)\\
-\,\sgn(x-y) & \ii \frac{k}{m+z}
\epm
\frac{\e^{\ii k |x-y|}}{2}.
\end{equation}

Now, let us denote by $H$ the restriction of $H_0$ to the domain
\[
\dom(H)=\big\{f\in \H^1(\R;\C^2)\mid f(0)=0\big\}.
\]
This operator is not self-adjoint any more, but symmetric with deficiency indices $(2,2)$, see for example \cite[Sec.~5.3]{BMN}.
Furthermore, its adjoint $H^*$ is given by the same expression but acts on the larger domain
$\dom(H^*)=\H^1(\R_-;\C^2)\oplus \H^1(\R_+;\C^2)$.
By a Dirac operator with a zero-range (or delta-type) interaction at the origin we mean any self-adjoint extension of $H$.
Note that there already exist several papers discussing various parameterizations of such extensions, see e.g.~\cite{alonso,BMN,dabr,cmp},
but our approach will be slightly different.

By a direct computation one can check that
$(\C^2,\Gamma_1,\Gamma_2)$ with
\[
\Gamma_1 f=\bpm
f_1(+0)-f_1(-0)\\
f_2(-0)-f_2(+0)

\epm,
\quad
\Gamma_2 f=\frac{1}{2}\,\bpm
f_2(-0)+f_2(+0)\\
f_1(-0)+f_1(+0)
\epm.
\]
is a boundary triple for $H^*$ in the sense that for all $f,g\in\dom (H^*)$ one has
\[
\langle f, H^*g\rangle-\langle H^*f, g\rangle
=\langle \Gamma_1 f,\Gamma_2 g\rangle -\langle \Gamma_2 f,\Gamma_1 g\rangle,
\]
and that the map $\dom(H^*)\ni f\mapsto (\Gamma_1 f,\Gamma_2 f)\in \C^2\times \C^2$ is surjective.
Here we have used the notation $f_j(\pm 0)$ for $\lim_{\varepsilon\searrow 0}f_j(\pm \varepsilon)$ for any $f_j \in \H^1(\R_\pm)$.

Let us now recall a few facts about boundary triples, and
refer to~\cite{BGP} for a brief account, to \cite{DM,gg} for a detailed study,
and to \cite{schm} for a recent textbook presentation.
One of the main interest of boundary triples is that they easily provide a simple description of all self-adjoint extensions of $H$
and a tool for their spectral and scattering analysis.
More precisely, let $C, D\in M_2(\C)$ be $2\times 2$ matrices, and let us denote by $H^{\CD}$  the restriction of $H^*$ to the domain
\begin{equation*}
\dom(H^{\CD}):=\{f \in \dom(H^*)\mid C\Gamma_1 f = D \Gamma_2 f\}\ .
\end{equation*}
Then, the operator $H^{\CD}$ is self-adjoint if and only if the matrices $C$ and $D$ satisfy the following conditions:
\begin{equation}
\label{eq-mcd}
\text{(i) $CD^*$ is self-adjoint,\qquad  (ii) $\det(CC^* + DD^*)\neq 0$.}
\end{equation}
Moreover, any self-adjoint extension of $H$ in $\H$ is equal to one of the operator $H^{\CD}$.
For simplicity, a pair $(C,D)$ of elements of $M_2(\C)$ satisfying relations \eqref{eq-mcd} will be called an admissible pair.
Note that with this notation, the operator $H^{10}\equiv H^*|_{\ker(\Gamma_1)}$ is exactly the above
free Dirac operator $H_0$, which is going to play the role of our reference operator. The other operators $H^\CD$ will be interpreted as its perturbations.

Let us stress that the above parametrization is not unique in the sense that one can have $H^{\CD}=H^{C'\!D'}$ for two different admissible pairs
$(C,D)$ and $(C',D')$. This is the case if and only if $C=KC'$ and $D=KD'$ for a non-degenerate $2\times 2$ matrix $K$.
One may obtain a one-to-one parametrization between the $2\times 2$ unitary matrices $U$ and the self-adjoint extensions of $H$
by setting $C:=\frac{1}{2}(1-U)$ and $D:=\frac{i}{2}(1+U)$. We refer to the papers \cite{alonso,dabr} for alternative one-to-one parameterizations.

\begin{remark}
We note that the papers \cite{BMN,cmp} also make use of boundary triples for studying perturbed Dirac operators, but
with a different choice for the maps $\Gamma_1$ and $\Gamma_2$, which leads to a rather complicated reference operator.
Our choice is motivated by obtaining simpler expressions in the subsequent computations.
\end{remark}

\subsection{Weyl function and resolvent formula}

Our aim is to obtain a resolvent formula for the self-adjoint extensions $H^\CD$. First, we compute explicitly
two operator-valued maps playing a key role in the framework of boundary triples, namely the map
\[
\gamma(z):=\big(\Gamma_1\big|_{\ker(H^*-z)}\big)^{-1}
\]
and the Weyl function $M(z):=\Gamma_2\gamma(z)$.
By a direct computation one obtains for any $z\notin\spec(H_0)$, $\xi = \big(\bsm \xi_1\\ \xi_2\esm \big)\in \C^2$ and $x \in \R^*$
that
\begin{equation*}
\left[\gamma(z)\bpm \xi_1 \\ \xi_2 \epm\right](x)=
\bpm
\ii\dfrac{m+z}{k} \,\xi_2+\xi_1\sgn (x)\\[\bigskipamount]
\ii \dfrac{k}{m+z}\,\xi_1-\xi_2\sgn(x)
\epm
\frac{\e^{\ii k|x|}}{2}\\
\equiv \xi_1\, h^1_z(x) +\xi_2\, h^2_z(x),
\end{equation*}
where we have set
\[
h^1_z(x)=\frac{\e^{\ii k|x|}}{2}
\bpm
\sgn (x)\\
\ii \frac{k}{m+z}
\epm,
\quad
h^2_z(x)=\frac{\e^{\ii k|x|}}{2}
\bpm
\ii\frac{m+z}{k}\\
-\,\sgn(x)
\epm\,.
\]
For later use, note that the $\gamma$-function satisfies the following identity \cite[Th.~1.23]{BGP}
\begin{equation}
    \label{gamh}
\gamma(\bar z)^*(H_0-z)f= \Gamma_2f, \quad f\in\dom(H_0), \quad z\notin\spec(H_0).
\end{equation}

Similarly, for the Weyl function we obtain
\begin{equation*}
M(z)\bpm \xi_1\\ \xi_2\epm
=
\frac{1}{2}\,
\bpm
\dfrac{\ii k}{m+z}\, \xi_1\\[\bigskipamount]
\dfrac{\ii (m+z)}{k}\, \xi_2
\epm,
\end{equation*}
{\it i.e.}~$M(z)$ is just the diagonal matrix
\[
M(z)=
\dfrac{1}{2} \;\!\diag\left(
\frac{\ii k}{m+z}, \frac{\ii (m+z)}{k}\right).
\]
In particular, since for $\lambda \in \R$ the following limits hold:
\begin{align*}
k(\lambda+\ii 0)\equiv
\lim_{\varepsilon\searrow 0} k(\lambda+i\varepsilon)&=\begin{cases}
\sgn(\lambda)\,\sqrt{\lambda^2-m^2}, & |\lambda|\ge m,\\
\ii\,\sqrt{m^2-\lambda^2}, & |\lambda|<m,
\end{cases}\\
k(\lambda-\ii 0)\equiv
\lim_{\varepsilon\searrow 0} k(\lambda-i\varepsilon)&=\begin{cases}
-\sgn(\lambda) \sqrt{\lambda^2-m^2}, & |\lambda|\ge m,\\
\ii\,\sqrt{m^2-\lambda^2}, & |\lambda|<m.
\end{cases}
\end{align*}
it follows that for $|\lambda|<m$ one has
$$
M(\lambda\pm i 0) =
\frac{1}{2}\, \diag\left(
-\sqrt{\frac{m-\lambda}{m+\lambda}}, \sqrt{\frac{m+\lambda}{m-\lambda}}\right)
$$
while for $|\lambda|>m$ one has
$$
M(\lambda\pm \ii 0) =
\pm \frac{\ii}{2}\diag\left(
\sqrt{\frac{\lambda-m}{\lambda+m}}, \sqrt{\frac{\lambda+m}{\lambda-m}}\right) .
$$

With these various definitions a few additional relations between the operators $H_0$ and $H^{\CD}$ can be inferred. For example,
for $z\notin \spec(H_0)\cup\spec(H^\CD)$ the matrix $DM(z)-C$ is invertible, and the resolvent formula
\begin{equation}
\label{eq-krein}
(H_0-z)^{-1}-(H^{\CD}-z)^{-1}=\gamma(z)\big(DM(z)-C\big)^{-1}D\gamma(\overline z)^*
\end{equation}
holds.
In addition, a value $\lambda\in(-m,m)$ is an eigenvalue
of $H^{\CD}$ if and only if $\det\big(DM(\lambda+i0)-C\big)=0$, and then one has $\ker(H^{\CD}-\lambda)=\gamma(\lambda)\ker(DM(\lambda+i0)-C)$.
Due to the injectivity of the map $\gamma(\lambda):\C^2 \to \H$, the dimension of $\ker(DM(\lambda+i0)-C)$ corresponds to the multiplicity of the eigenvalue $\lambda$ of $H^{\CD}$.
Finally, for any $\varepsilon\ge 0$, any $\lambda\not \in [-m,m]$ and any admissible pair $(C,D)$, let us mention the obvious equality $M(\lambda-i\epsilon)=M(\lambda+i\varepsilon)^*$
and the identity (see for example \cite[Lem.~6]{PR})~:
\begin{equation}\label{tadj}
\Big[\big(DM(\lambda-i\varepsilon)-C\big)^{-1}D\Big]^*=\big(DM(\lambda+i\varepsilon)-C\big)^{-1}D.
\end{equation}

For the next statement, we need to introduce the set
$$
\Sigma:= (-\infty,-m)\,\cup\,(m,+\infty).
$$
as well as for each $\lambda \in \Sigma$ the $2\times 2$ matrix
\[
B(\lambda) = \frac{1}{\sqrt{2}}\;\!
\diag\left(
\sqrt[4]{\frac{\lambda-m}{\lambda+m}}, \sqrt[4]{\frac{\lambda+m}{\lambda-m}}\right)
\]
which clearly satisfies $iB(\lambda)^2 = M(\lambda+i0)$.
Furthermore, we also set
\[
T_\varepsilon^{\CD}(\lambda):=-2iB(\lambda) \big(DM(\lambda+i\varepsilon)-C\big)^{-1} D\;\!B(\lambda).
\]
We summarize some properties of this operator in the following assertion, whose proof is given
in the Appendix.

\begin{lemma}\label{lem-unt}
For any admissible pair $(C,D)$, the operator $T_\varepsilon^{\CD}(\lambda)$ admits the limit
\begin{equation}\label{defT0}
T_0^{\CD}(\lambda)=-2iB(\lambda) \big(DM(\lambda+i 0)-C\big)^{-1} D\;\!B(\lambda)
\end{equation}
in $M_2(\C)$ as $\varepsilon\searrow0$ locally uniformly
in $\lambda\in\Sigma$. For any $\lambda\in\Sigma$ the matrix $1+T_0^{\CD}(\lambda)$ is unitary,
and the map $\Sigma \ni \lambda \mapsto T_0^{\CD}(\lambda)\in M_2(\C)$
is continuous and admits limits at the boundary points of $\Sigma$, with $T_0^{\CD}(-\infty)=T_0^{\CD}(+\infty)$.
\end{lemma}

\subsection{Spectral analysis}

In the next statement, we infer some spectral results for the operator $H^{\CD}$.

\begin{prop}\label{propspectral}
Let $(C,D)$ be an admissible pair. Then, one has $\spec_{ac}(H^{\CD})=(-\infty,-m]\cup[m,+\infty)$
and $\spec_{sc}(H^{\CD})=\emptyset$. Moreover, $\spec_{p}(H^{\CD})\subset (-m,m)$,
and the number of eigenvalues can be explicitly described as follows:
\begin{enumerate}
\item[(a)] If $\det (D)\ne 0$ and
$D^{-1}C:=\bpm
\ell_{11} & \ell_{12}\\
\overline \ell_{12} & \ell_{22}
\epm$,
then
\[
\#\spec_p (H^{\CD})=\begin{cases}
2 & \text{if  $\ell_{11}<0$ and $\ell_{22}>0$},\\
0 & \text{if  $\ell_{11}\ge 0$ and $\ell_{22}\le0$},\\
1 & \text{otherwise}.
\end{cases}
\]
\item[(b)] If $\dim[\ker (D)]= 1$ and $(p_1,p_2)$ is a unit vector spanning $\ker (D)$, then $\#\spec_p (H^{\CD})=1$ if
$p_1p_2\neq 0$, if $p_2 =0$ and $\tr(CD^*)>0$, or if $p_1=0$ and $\tr(CD^*)<0$. In the other cases, $\#\spec_p (H^{\CD})=0$.
\item[(c)] If $D=0$, then $\spec_p (H^{\CD})=\emptyset$.
\end{enumerate}
\end{prop}

\begin{proof}

Since the difference of the resolvents of $H^{\CD}$ and $H_0$ is a finite rank operator,
it clearly follows that $\spec_{ess}(H^{\CD})=\spec_{ess}(H_0)=\spec(H_0) = (-\infty,-m]\cup[m,+\infty)$.

Let $f\in \H$ and let $\mu^{\CD}_f$ be the spectral measure associated with $H^{\CD}$ and $f$.
It is well known (see for example \cite[Thm.~4.15]{Jak}) that the singular part $\mu^{\CD}_{f,s}$ of $\mu^{\CD}_f$ is concentrated on the set
\begin{equation*}
\big\{
\lambda\in\R \mid \lim_{\varepsilon \searrow 0}\Im \big\langle f, (H^{\CD}-\lambda-i\varepsilon)^{-1}f\big\rangle=\infty
\big\}.
\end{equation*}
So, let us consider $f\in C_c^\infty(\R;\C^2)$ and compute the above quantity for any $|\lambda|>m$. By \eqref{eq-krein} and for $z=\lambda +i\varepsilon$, this reduces to an evaluation of the expressions
\begin{equation}\label{2_termes}
\lim_{\varepsilon \searrow 0}\Im \big\langle f, (H_0-z)^{-1}f\big\rangle \quad \hbox{ and } \quad \lim_{\varepsilon \searrow 0}\Im \big\langle \gamma(z)^*f, \big(DM(z)-C\big)^{-1}D\gamma(\overline z)^*f\big\rangle_{\C^2}.
\end{equation}
A simple computation which takes the expression \eqref{Green_f} for the Green function into account shows that the first term in \eqref{2_termes} is finite. Similarly, by using the explicit expression for $\gamma(z)$, one easily obtains that the limits $\lim_{\varepsilon\searrow 0} \gamma(\lambda\pm i\varepsilon)^*f\in \C^2$ exist.
On the other hand, still for $|\lambda|> m$ one has $\Im M(\lambda+i 0)>0$, hence, by \cite[Lem.~6]{PR}, the limit $\big(D M(\lambda+i0)-C\big)^{-1}\in M_2(\C)$ exists.
Therefore, the second expression in \eqref{2_termes} is also finite, and thus the support of $\mu^{\CD}_{f,s}$ does not intersect
the set $(-\infty,-m)\cup(m,+\infty)$. Since $f$ is an arbitrary element of a dense set in $\H$, it means that $H^{\CD}$ has no singular spectrum in $(-\infty,-m)\cup(m,+\infty)$, and in particular that the singular continuous spectrum is empty.

Now, to see that $\pm m\notin \spec_p (H^{\CD})$
it is sufficient to observe that the only solutions of the ordinary differential equation
$H f = \pm m f$ with $H$ given by the expression~\eqref{eq-dirac} are either constant or polynomially growing. In particular, this equation has no solution in $\H$.

It remains to count the eigenvalues of $H^{\CD}$ in the interval $(-m,m)$.
For that purpose and as noted in the previous section, we first need to determine if the operator $DM(\lambda+i0)-C$ has a $0$-eigenvalue for some $\lambda \in (-m,m)$.
To simplify the notation, let us set
$t=\sqrt{\frac{m-\lambda}{m+\lambda}}$.
Then, this problem reduces in the study of the possible $0$-eigenvalue of the operator $\frac{1}{2}D \;\!\diag( -t,1/t) -C$ for $t \in (0,\infty)$.

a) We first consider the case $\det(D)\ne 0$. For that purpose, we set $\Lambda=D^{-1}C$ and study the equivalent admissible pair $(\Lambda,1)$. Now, observe that the matrix $\Lambda$ is hermitian:
\[
\Lambda=:\bpm
\ell_{11} & \ell_{12}\\
\overline \ell_{12} & \ell_{22}
\epm.
\]
Thus, we are left with the study of the determinant of the operator $\frac{1}{2}\;\!\diag(-t,1/t)-L$ with $t \in (0,\infty)$.
Let us still define the map
\[
(0,\infty)\ni t \mapsto D_\Lambda(t):={\textstyle \frac{1}{2}\ell_{22} t -\frac{1}{2}\ell_{11}/t + \det(\Lambda) -\frac{1}{4} \in \R}
\]
Clearly, the determinant of the mentioned operator vanishes for some $t \in (0,\infty)$ if and only if the equation $D_\Lambda(t)=0$. Now, if $\ell_{11}\ell_{22}>0$, then the map $D_\Lambda(\cdot)$
has no local extremum and thus the above equation has always one single solution. On the other hand, if $\ell_{11} \ell_{22}<0$, then this equation may have $0$, $1$ or $2$ solutions. Indeed, in that case the map $D_\Lambda(\cdot)$ takes its local extremum at the value $t = \sqrt{-\ell_{11}/\ell_{22}}$. Then, $D_\Lambda(\cdot)$ vanishes only once if $\sgn(\ell_{11})\;\!\big[\det(L)-\frac{1}{4}\big] =  \sqrt{-\ell_{11}\ell_{22}}$, $D_\Lambda(\cdot)$ vanishes twice on $(0,\infty)$ if
$\sgn(\ell_{11})[\det(L)-\frac{1}{4}]>\sqrt{-\ell_{11}\ell_{22}}$, while $D_\Lambda(\cdot)$ does not vanish in the remaining case.
However, note that there is a very explicit set of solutions of the relation $\sgn(\ell_{11})\;\!\big[\det(L)-\frac{1}{4}\big] =  \sqrt{-\ell_{11}\ell_{22}}$, namely when $\ell_{11}<0$, $\ell_{11}\ell_{22}=-\frac{1}{4}$ and $\ell_{12}=0$.
In addition, a simple computation shows that the two conditions $\ell_{11} \ell_{22}<0$ and
$\sgn(\ell_{11})[\det(L)-\frac{1}{4}]>\sqrt{-\ell_{11}\ell_{22}}$ hold whenever $\ell_{11}<0$, $\ell_{22}>0$, and either $\ell_{11}\ell_{22}\neq -\frac{1}{4}$ or $\ell_{12}\neq 0$, while the two conditions $\ell_{11} \ell_{22}<0$ and
$\sgn(\ell_{11})[\det(L)-\frac{1}{4}]<\sqrt{-\ell_{11}\ell_{22}}$
hold whenever $\ell_{11}>0$ and $\ell_{22}<0$.
Now, if $\ell_{11}=0$ then $D_\Lambda(\cdot)$ vanishes once on $(0,\infty)$ if $\ell_{22} >0$ and does not vanish if $\ell_{22}<0$; if $\ell_{22}=0$ and $\ell_{11} < 0$ then $D_\Lambda(\cdot)$ vanishes once while if $\ell_{11} > 0$ then $D_\Lambda(\cdot)$ does not vanish on $(0,\infty)$.
Finally, if $\ell_{11}=\ell_{22}=0$ there is no solution for $D_\Lambda(\cdot)=0$.

b) We now consider the case $\dim[\ker (D)]=1$ and follow the construction described in \cite[Sec.~3]{PR}. Let $(p_1,p_2)$ be a unit vector spanning $\ker (D)$.
Let $I:\C \to \C^2$ be the identification of $\C$ with $\ker(D)^\bot$, and let $P$ denote its adjoint, {\it i.e.}~$P:\C^2 \to \C$ is the composition of the orthogonal projection onto $\ker(D)^\bot$ together with the identification of $I \C$ with $\C$.
Then, as shown in \cite[Eq. (12)]{PR}, the operator $DM(\lambda+i0)-C$ is invertible if and only if the reduced operator $PM(\lambda+i0)I-\ell$ is invertible, where $\ell := (DI)^{-1}CI \in \R$.
By using this observation and the change of variable $t=\sqrt{\frac{m-\lambda}{m+\lambda}}$
we see that the $0$-eigenvalue of the operator $\frac{1}{2}D \;\!\diag( -t,1/t) -C$ for $t \in (0,\infty)$
coincides with the $0$ of the map
\[
(0,+\infty)\ni t\mapsto d_\ell(t):=|p_2|^2 t - |p_1|^2 / t +2\ell \in \R.
\]
By elementary considerations we observe that this map vanishes once if $p_1p_2\neq 0$, if $p_2 =0$ and $\ell>0$, or if $p_1=0$ and $\ell<0$. In the other cases the map $d_\ell(\cdot)$ never vanishes.
Finally, in order to obtain the statement of the lemma, let us recall a useful relation between $\ell$ and $(C,D)$, namely $\sgn (\ell)= \sgn[\tr (CD^*)]$, see \cite[Sec.~3.B]{KPR}.

c) If $D=0$, then $H^{\CD}=H_0$ and hence $H^{\CD}$ has no eigenvalue.

It now only remains to relate these various results with the dimension of $\H_p(H^{\CD})$. As mentioned before the statement of the lemma, for each $\lambda\in (-m,m)$ such that $0$ is an eigenvalue of $DM(\lambda+i0)-C$, one needs to determine the multiplicity of this eigenvalue. Equivalently, for each $t\in (0,\infty)$ such that $0$ is an eigenvalue of $\frac{1}{2}D \;\!\diag( -t,1/t) -C$, one needs to determine the multiplicity of this eigenvalue. An simple inspection in the previous computations shows that the multiplicity of the $0$-eigenvalue is always $1$, except in one special case already emphasized above, namely when $\ell_{11}<0$, $\ell_{11}\ell_{22}=-\frac{1}{4}$ and $\ell_{12}=0$, for which the multiplicity if $2$. By collecting all these results, one finally obtains the statement of the lemma.
\end{proof}

\section{Scattering theory}\label{secscat}
\setcounter{equation}{0}

\subsection{Wave operators and scattering operator}

In this section, we describe the scattering theory for the pair of the operators $(H^\CD,H_0)$.
Since $\big(H^{\CD}-i\big)^{-1} - (H_0-i)^{-1}$ is a finite dimensional operator, it is well-known that the time dependent wave operators
\begin{equation}\label{defW}
W_\pm\big(H^{\CD},H_0\big):=s-\lim_{t \to \pm \infty} \e^{itH^{\CD}}\e^{-itH_0}
\end{equation}
exist and are complete, see for example \cite[Thm.~6.5.1]{Y}.
Then, the operator
\begin{equation}\label{defS}
S\big(H^{\CD},H_0\big):=W_+\big(H^{\CD},H_0\big)^*W_-\big(H^{\CD},H_0\big)
\end{equation}
is usually referred to as the scattering operator. Thus, the aim of the present section
is to calculate these objects in terms of the Weyl function
and of the parameters $(C,D)$. To do that, we recall
the so-called stationary expressions for the wave operators.
Namely, for suitable $f,g \in \H$ we consider the operators $W_\pm^{\CD}$
defined by
\begin{equation*}
\langle W_\pm^{\CD}f,g\rangle = \int_\R \lim_{\varepsilon \searrow 0}\frac{\varepsilon}{\pi}\big\langle (H_0-\lambda\mp i\varepsilon)^{-1}f,(H^{\CD}-\lambda\mp i\varepsilon)^{-1}g\big\rangle \;\!\d \lambda.
\end{equation*}
Note that the precise choice for the elements $f,g$ will be specified later on, and that the equality of $W_\pm\big(H^{\CD},H_0\big)$ with $W_\pm^{\CD}$ will follow from our computations.
In the sequel, we concentrate on $W_-^\CD$, and stress that the operator $W_+^\CD$ can be treated similarly.

For that purpose, consider for $\varepsilon>0$ the function $\delta_\varepsilon:\R\to\R$ given by
\[
\delta_\varepsilon (x)=\dfrac{1}{\pi}\dfrac{\varepsilon}{x^2+\varepsilon^2}.
\]
We clearly have
\[
\delta_\varepsilon\big(H_0-\lambda\big)
=\frac{\varepsilon}{\pi}(H_0-\lambda+i \varepsilon)^{-1}(H_0-\lambda-i \varepsilon)^{-1}.
\]
With this notation, the limit
$
\lim_{\varepsilon\searrow 0}
\big\langle\delta_\varepsilon\big(H_0-\lambda\big)f,g\big\rangle
$
exists for a.e. $\lambda\in\R$, and
\begin{equation*}
\int_\R\lim_{\varepsilon\searrow0}
\big\langle\delta_\varepsilon(H_0-\lambda)f,g\big\rangle \;\!\d\lambda
= \langle f,g\rangle, 
\end{equation*}
see \cite[Sec.~1.4]{Y}. As a consequence, by taking the resolvent formula \eqref{eq-krein} into account, one obtains that
\begin{multline}
       \label{ww1}
 \langle (W_-^{\CD}-1)f,g\rangle \\
=  -\int_\R\lim_{\varepsilon \searrow 0}
\Big\langle \frac{\varepsilon}{\pi}\gamma(\lambda -i\varepsilon)^* (H_0-\lambda +i\varepsilon)^{-1} f,  \big(DM(\lambda-i\varepsilon)-C\big)^{-1} D
\gamma(\lambda +i\varepsilon)^* g\Big\rangle_{\C^2}
\d \lambda.
\end{multline}
Our next aim will be to put together some of these terms and to obtain some more coherent and simply understandable factors. For that purpose, the spectral representation of the operator $H_0$ will be needed.

\subsection{Spectral representation of the free Dirac operator}

In this section, we construct the spectral representation of $H_0$, mimicking the construction provided in \cite[Sec.~2]{Ito}.
For any fixed $p \in \R$, let us set
\[
h(p) := \bpm m & -ip \\ ip & -m\epm \ \in M_2(\C).
\]
The eigenvalues of this matrix are $\pm \sqrt{p^2+m^2}$ and two normalized eigenfunction are defined by the expressions
\begin{equation}
         \label{xixi}
\begin{gathered}
\xi^+(p):=
\dfrac{1}{\sqrt{2(p^2+m^2+m\sqrt{p^2+m^2})}}\,
\bpm
m+\sqrt{p^2+m^2} \\ i p
\epm,\\
\xi^-(p):=
\dfrac{1}{\sqrt{2(p^2+m^2+m\sqrt{p^2+m^2})}}\,
\bpm
i p \\ m+\sqrt{p^2+m^2}
\epm.
\end{gathered}
\end{equation}
For simplicity, the orthogonal projection on the subspace generated by $\xi^\pm(p)$ will be denoted by $P^\pm(p)\in M_2(\C)$.
Then, for any $\lambda \in \R$ satisfying $\pm \lambda >m$, let us define
\begin{equation*}
\Hrond(\lambda) := \Big( P^\pm(-\sqrt{\lambda^2-m^2}\big)\;\!\C^2, P^\pm\big(\sqrt{\lambda^2-m^2}\big)\;\!\C^2\Big),
\end{equation*}
and observe that $\Hrond(\lambda)$ is a two dimensional subspace of $ \C^2\oplus \C^2$.
Let us also set
\begin{equation*}
\Hrond:=\int_\Sigma^\oplus \Hrond(\lambda)\;\! \d \lambda.
\end{equation*}
More precisely, any element $\varphi \in \Hrond$ is of the form
$(\varphi_1,\varphi_2)$ with $\varphi_j \in L^2(\Sigma; \C^2)$,
$\varphi_1(\lambda)$ collinear to $\xi^\pm\big(-\sqrt{\lambda^2-m^2}\big)$ and
$\varphi_2(\lambda)$ collinear to $\xi^\pm\big(\sqrt{\lambda^2-m^2}\big)$ for $\pm \lambda>m$.
One also defines the unitary operator $\U : \H \to \Hrond$ given for $f \in \H$ and $\pm \lambda>m$ by
\begin{equation*}
[\U f](\lambda):= \sqrt[4]{\frac{\lambda^2}{\lambda^2-m^2}}
\Big( P^\pm(-\sqrt{\lambda^2-m^2}\big)\;\!f\big(-\sqrt{\lambda^2-m^2}\big), P^\pm\big(\sqrt{\lambda^2-m^2}\big)\;\!f\big(\sqrt{\lambda^2-m^2}\big)\Big).
\end{equation*}
Note that its adjoint is provided by the following expression: for any $\varphi \in \Hrond$ with $\varphi = (\varphi_1,\varphi_2)$ and for $p \in \R^*$ one has
\begin{equation*}
[\U^*\varphi](p) =\sqrt[4]{\frac{p^2}{p^2+m^2}}
\begin{cases}
\varphi_1\big(-\sqrt{p^2 + m^2}\big) + \varphi_1\big(\sqrt{p^2 + m^2}\big)& \ \ \text{ if } p<0,\\
\varphi_2\big(-\sqrt{p^2 + m^2}\big) + \varphi_2\big(\sqrt{p^2 + m^2}\big)& \ \ \text{ if } p>0.
\end{cases}
\end{equation*}
Obviously, the above expressions have to be understood in the $L^2$-sense, {\it i.e.}~for almost every $\lambda \in \Sigma$ or for almost every $p \in \R^*$.
It is now a matter of a simple computation to check that for any $\lambda\in \Sigma$ one has
\begin{equation*}
[\U \;\!h(\cdot)\;\!\U^* \varphi](\lambda) = \lambda \;\! \varphi(\lambda).
\end{equation*}
In addition, if $\F$ denotes the Fourier transform on $L^2(\R; \C^2)$, then the operator $\F_0:= \U \F: \H \to \Hrond$ realizes the spectral representation of $H_0$, namely
\begin{equation}\label{rel1}
\F_0\;\!H_0\;\!\F_0^* = L_0,
\end{equation}
where $L_0$ is the self-adjoint operator of multiplication by the variable $\lambda$ in $\Hrond$.

\subsection{Computing the wave operator : preliminary steps}

In order to simplify the expression~\eqref{ww1} we first use the identity \eqref{gamh}, which gives,
for $z\notin\spec(H_0)$,
\begin{align*}
\frac{\varepsilon}{\pi}
\gamma(\bar z)^*(H_0-\bar z)^{-1}
&= \frac{\varepsilon}{\pi}\,
\Gamma_2(H_0-z)^{-1}(H_0-\bar z)^{-1}
= \Gamma_2\;\! \delta_\varepsilon(H_0-\lambda),\\
\gamma(z)^* &= \Gamma_2\;\!(H_0-\bar z)^{-1}.
\end{align*}
By collecting these equalities
and using \eqref{tadj}, one infers that
\begin{align*}
\langle (W_-^{\CD}-1)f,g\rangle
& = -\int_\R\lim_{\varepsilon \searrow 0}
\big\langle  \Gamma_2 \delta_\varepsilon(H_0-\lambda)f,
\big(DM(\lambda-i\varepsilon)-C\big)^{-1} D
\Gamma_2 (H_0-\lambda+ i\varepsilon)^{-1} g\big\rangle_{\C^2}\;\!
\d \lambda \\
& = \frac{1}{2} \int_\Sigma \lim_{\varepsilon \searrow 0}
\big\langle  T_\varepsilon^{\CD}(\lambda) B(\lambda)^{-1} \Gamma_2 \delta_\varepsilon(H_0-\lambda)f,
i B(\lambda)^{-1}\Gamma_2 (H_0-\lambda+ i\varepsilon)^{-1} g\big\rangle_{\C^2}\;\!
\d \lambda.
\end{align*}
Note that for the second equality, we have taken into account that the above
integrant vanishes for almost every $\lambda \in (-m,m)$ as $\varepsilon \searrow 0$;
more precisely, it vanishes at any point $\lambda\in(-m,m)$ which is not an eigenvalue
of $H^\CD$.

By looking at the previous equality inside the spectral representation of $H_0$, one has thus obtained that for suitable $\varphi, \psi \in \Hrond$:
\begin{multline}
 \langle \F_0(W_-^{\CD}-1)\F_0^*\varphi, \psi\rangle_{_{\!\!\Hrond}} \\
=   \frac{1}{2}\int_\Sigma \lim_{\varepsilon \searrow 0}
\big\langle  T_\varepsilon^{\CD}(\lambda) B(\lambda)^{-1} \Gamma_2 \F_0^* \delta_\varepsilon(L_0-\lambda)\varphi,
i B(\lambda)^{-1}\Gamma_2 \F_0^*(L_0-\lambda+ i\varepsilon)^{-1} \psi \big\rangle_{\C^2}\;\!
\d \lambda.
   \label{ww2}
\end{multline}
For the next statement, one needs to be a little bit more cautious about the set of suitable elements of $\Hrond$.
For that purpose, we introduce the following space:
\begin{equation*}
\SS:= \Big\{\eta = \bpm \eta_1 \\ \eta_2\epm \mid  \eta_j \in C_c^\infty(\Sigma) \Big\}.
\end{equation*}
Our interest in this set comes from its dense embedding into $\Hrond$.
Indeed, define
\[
J: \SS \to \Hrond, \quad
[J\eta](\lambda):=\Big(\eta_1(\lambda)\;\!\xi^\pm\big(-\sqrt{\lambda^2-m^2}\big) ,
\eta_2(\lambda)\;\!\xi^\pm\big(\sqrt{\lambda^2-m^2}\big)\Big)
\text{ for } \pm \lambda >m.
\]
It clearly follows that $J\SS$ is dense in $\Hrond$ and that $J$ extends to a unitary operator from $L^2(\Sigma;\C^2)$ to $\Hrond$.
We then set
\begin{equation}\label{defL}
L:=J^*L_0 J,
\end{equation}
i.e. $L$ is simply the operator of multiplication by the variable in $L^2(\R,\C^2)$.

Let us finally introduce for each $\lambda \in \Sigma$ the unitary $2\times 2$ matrix $N(\lambda)$ defined by
\[
N(\lambda) =\frac{1}{\sqrt{2}}\bpm 1 & 1 \\ -i & i \epm \hbox{ if } \lambda < -m
\quad \hbox{ and }\quad
N(\lambda) =\frac{1}{\sqrt{2}}\bpm -i & i \\ 1 & 1 \epm \hbox{ if } \lambda > m.
\]
The operator of multiplication by the function $N$ defines a unitary operator in $L^2(\Sigma;\C^2)$,
and it will be denoted by the same symbol $N$.

\begin{remark}
Let us stress that the precise form of this unitary transformation $N$ is not really relevant here. Indeed, this transformation highly depends on our choice for the functions $\xi^\pm$ in the spectral representation of $H_0$. In fact, only the product $NJ^*\U$ really matters, as it can be inferred from formula \eqref{egalite} (recall that $\F_0 = \U\F$).
\end{remark}

\begin{lemma}\label{lem5}
For any $\eta \in \SS$ and $\lambda \in \Sigma$ one has
\begin{equation*}
\lim_{\varepsilon \searrow 0}  \sqrt{\pi} B(\lambda)^{-1} \Gamma_2 \F_0^* \delta_\varepsilon(L_0-\lambda)J \eta
= N(\lambda)\;\!\eta(\lambda).
\end{equation*}
\end{lemma}

\begin{proof}
For shortness, let us set $p_m := \sqrt{p^2+m^2}$. Then, a simple computation
gives
\begin{align*}
& \Gamma_2 \F_0^* \delta_\varepsilon(L_0-\lambda)J \eta \\
= & \frac{1}{\sqrt{2\pi}}
\int_{-\infty}^0 \sqrt[4]{\frac{p^2}{p^2+m^2}}
\Big[\delta_\varepsilon(-p_m-\lambda) \eta_1(-p_m) \xi^+(p) +
\delta_\varepsilon(p_m-\lambda) \eta_1(p_m) \xi^-(p)\Big] \d p \\
& + \frac{1}{\sqrt{2\pi}}
\int^{\infty}_0 \sqrt[4]{\frac{p^2}{p^2+m^2}}
\Big[\,\delta_\varepsilon(-p_m-\lambda) \eta_2(-p_m) \xi^+(p) +
\delta_\varepsilon(p_m-\lambda)\, \eta_2(p_m) \xi^-(p)\Big] \d p \\
= & \frac{1}{\sqrt{2\pi}} \int_{-\infty}^{-m}
\sqrt[4]{\frac{\mu^2}{\mu^2-m^2}}\, \delta_\varepsilon(\mu-\lambda)\, \Big[\eta_1(\mu) \xi^+\big(-\sqrt{\mu^2 - m^2}\big) +\eta_2(\mu) \xi^+\big(\sqrt{\mu^2 - m^2}\big)\Big] \d \mu
\\
& + \frac{1}{\sqrt{2\pi}} \int^{\infty}_{m}
\sqrt[4]{\frac{\mu^2}{\mu^2-m^2}}\, \delta_\varepsilon(\mu-\lambda) \Big[ \eta_1(\mu) \xi^-\big(-\sqrt{\mu^2 - m^2}\big)+ \eta_2(\mu)  \xi^-\big(\sqrt{\mu^2 - m^2}\big) \Big] \d \mu.
\end{align*}
By taking the limit as $\varepsilon \searrow 0$ one obtains for $\pm \lambda >m$
\begin{equation*}
\lim_{\varepsilon \searrow 0 } \Gamma_2 \F_0^* \delta_\varepsilon(L_0-\lambda)J\eta
=\frac{1}{\sqrt{2\pi}}
\sqrt[4]{\frac{\lambda^2}{\lambda^2-m^2}}
\Big[ \eta_1(\lambda) \xi^{\mp}\big(-\sqrt{\lambda^2 - m^2}\big)+ \eta_2(\lambda) \xi^{\mp}\big(\sqrt{\lambda^2 - m^2}\big)\Big].
\end{equation*}
Finally, by using the expressions \eqref{xixi} for $\xi^\pm$ and by considering separately the case $\pm \lambda>m$,
a short computation leads directly to the statement.
\end{proof}

In addition, one can also show:

\begin{lemma}\label{lem6}
For any $\eta \in \SS$ and $\lambda \in \Sigma$ one has
$$
B(\lambda)^{-1}\Gamma_2 \F_0^*(L_0-\lambda+ i\varepsilon)^{-1} J\eta = \frac{1}{\sqrt{\pi}} \int_\Sigma
B(\lambda)^{-1} (\mu-\lambda+i\varepsilon)^{-1} B(\mu) \;\!N(\mu)\;\!\eta(\mu) \;\!\d \mu.
$$
\end{lemma}

\begin{proof}
By a computation similar to the previous proof, one obtains that
\begin{align*}
& \Gamma_2 \F_0^* (L_0-\lambda+i\varepsilon)^{-1}J \eta \\
= & \frac{1}{\sqrt{2\pi}} \int_{-\infty}^{-m}
\sqrt[4]{\frac{\mu^2}{\mu^2-m^2}} \;\!(\mu-\lambda+i\varepsilon)^{-1}  \Big[\eta_1(\mu) \xi^+\big(-\sqrt{\mu^2 - m^2}\big) +\eta_2(\mu) \xi^+\big(\sqrt{\mu^2 - m^2}\big)\Big] \d \mu
\\
& + \frac{1}{\sqrt{2\pi}} \int^{\infty}_{m}
\sqrt[4]{\frac{\mu^2}{\mu^2-m^2}} \;\!(\mu-\lambda+i\varepsilon)^{-1}  \Big[ \eta_1(\mu) \xi^-\big(-\sqrt{\mu^2 - m^2}\big)+ \eta_2(\mu)  \xi^-\big(\sqrt{\mu^2 - m^2}\big) \Big] \d \mu \\
= & \frac{1}{\sqrt{2\pi}} \int_{-\infty}^{-m}
(\mu-\lambda+i\varepsilon)^{-1}
\frac{1}{\sqrt{2}}\;\! \diag\left(
\sqrt[4]{\frac{\mu-m}{\mu+m}}, \sqrt[4]{\frac{\mu+m}{\mu-m}}\right)
\bpm 1 & 1 \\ -i & i \epm
\bpm \eta_1(\mu)\\ \eta_2(\mu)
\epm \d \mu
\\
& + \frac{1}{\sqrt{2\pi}} \int^{\infty}_{m}
(\mu-\lambda+i\varepsilon)^{-1} \frac{1}{\sqrt{2}}\;\! \diag\left(
\sqrt[4]{\frac{\mu-m}{\mu+m}}, \sqrt[4]{\frac{\mu+m}{\mu-m}}\right)
\bpm -i & i \\ 1 & 1 \epm
\bpm \eta_1(\mu)\\ \eta_2(\mu)
\epm \d \mu \\
= & \frac{1}{\sqrt{\pi}} \int_\Sigma
(\mu-\lambda+i\varepsilon)^{-1} B(\mu)\;\!N(\mu) \;\!\eta(\mu) \;\!\d \mu
\end{align*}
which leads directly to the expected result.
\end{proof}

For $\varepsilon >0$, let us finally define the integral operator $\Theta_\varepsilon$ on $\SS$ which kernel is
\begin{equation*}
\Theta_\varepsilon(\lambda,\mu):= \frac{i}{2\pi}
\ B(\lambda)^{-1} (\mu-\lambda+i\varepsilon)^{-1} B(\mu).
\end{equation*}
A straightforward computation leads then to the following equality for any
$\eta \in \SS$ and $\lambda \in \Sigma$:
\[
\big[\lim_{\varepsilon \searrow 0} \Theta_\varepsilon \eta\big](\lambda)
\equiv [\Theta_0\eta](\lambda)
=  \frac{i}{2\pi}  B(\lambda)^{-1} \
\Pv \!\!\int_\Sigma \frac{1}{\mu-\lambda}\;\! B(\mu) \;\!\eta(\mu)\;\!\d \mu + \frac{1}{2} \eta(\lambda).
\]
Starting from \eqref{ww2} and by taking Lemmas \ref{lem5} and~\ref{lem6} into account, one can already guess that the singular integral operator $\Theta_0$ is going to play a central role in the expression for the wave operator.
However, let us observe that the maps $\lambda \mapsto B(\lambda)$ and $\lambda \mapsto B(\lambda)^{-1}$ are not bounded as $\lambda \to \pm m$. Therefore, it is not very easy to deal with the above kernel. For that reason, our last task is to get a better understanding of this integral operator by looking at it in another unitarily equivalent representation.

\subsection{The upside-down representation}

Let us finally define the unitary operator $\V: L^2(\Sigma;\C^2) \to L^2(\R;\C^2)$ given for $\eta \in L^2(\Sigma;\C^2)$ and $x \in \R$ by
\[
\big[\V \eta\big](x) :=\sqrt{2m} \;\!\frac{\e^{x /2}}{\e^{x}-1} \, \eta\Big(m\frac{\e^{x}+1}{\e^{x}-1}\Big)
\]
The special feature of this representation is that the values $\pm m$ are sent to $\pm \infty$
while any neighbourhood of the points $\pm \infty$ is then located near the point $0$.
The adjoint of the operator $\V$ is provided for $\zeta \in L^2(\R;\C^2)$ and $\lambda \in \Sigma$ by the expression
\[
\big[\V^*\zeta\big](\lambda) = \sqrt{2m} \;\!\sqrt{\frac{\lambda+m}{\lambda-m}}\;\!\frac{1}{\lambda+m}\;\!
\,\zeta\Big(\ln\Big[\frac{\lambda+m}{\lambda-m}\Big]\Big).
\]
We shall now compute the kernel of the operator $\V \Theta_0 \V^*$, and observe that this new kernel has a very simple form.

For that purpose, let us use the standard notation $\X$ for the self-adjoint operator on $L^2(\R)$ of multiplication by the variable, and by $\D$ the self-adjoint operator on the same space corresponding to the formal expression $-i\frac{\d}{\d x}$.
For a measurable (matrix-valued) function $K:\R\to M_2(\C)$ we denote by $K(\X)$ the operator of pointwise multiplication
by the matrix $K(\cdot)$ in $L^2(\R;\C^2)$, and by $K(\D)$ we denote the operator $\F^* K(\X)\F$,
where $\F$ is the Fourier transformation in $L^2(\R;\C^2)$.

One checks, by a direct substitution, that for any measurable function $\rho:\Sigma\to M_2(\C)$ one has
\begin{equation}
        \label{vv1}
\V\rho(L)\V^*=\rho\Bigg(m\frac{\e^\X+1}{\e^\X-1}\Bigg).
\end{equation}
In particular, such a relation holds for $\rho=T_0^\CD$.
Furthermore, for any $\zeta=(\zeta_1,\zeta_2) \in C_c^\infty(\R\setminus\{0\};\C^2)$ and $x \in \R^*\setminus \{0\}$,
it can be obtained straightforwardly that
\begin{align*}
& [\V \Theta_0 \V^*\zeta](x) \\
=& \frac{i}{8\pi} \Pv \!\!\int_\R
\bpm
\frac{-1}{\sinh((y-x)/4)} + \frac{1}{\cosh((y-x)/4)} & 0 \\
0 & \frac{-1}{\sinh((y-x)/4)} + \frac{-1}{\cosh((y-x)/4)}
\epm
\zeta(y)\;\!\d y + \frac{1}{2} \zeta(x)\\
= & \frac{i}{8\pi}\bpm g_+\star \,\zeta_1 \\ g_- \star\, \zeta_2\epm (x) + \frac{1}{2}\, \zeta(x),
\end{align*}
where
\[
g_\pm(x)=\dfrac{1}{\sinh (x/4)} \pm \dfrac{1}{\cosh (x/4)},
\]
and $\star$ means the (distributional)
convolution product. Using then the identity
\[
g\star f =[\F g](\D)f
\]
and the explicit expressions for the Fourier images of $g_\pm$ from \cite[Table 20.1]{Jef}
we obtain $\V \Theta_0 \V^* = R(\D)^*$
with $R(\cdot)$ defined for all $x \in \R$ by
\begin{equation}\label{defR}
R(x):=\frac{1}{2} \Bigg[
\bpm
\tanh(2\pi x) - i\cosh(2\pi x)^{-1} & 0 \\
0 & \tanh(2\pi x) +i \cosh(2\pi x)^{-1}
\epm + 1
\Bigg],
\end{equation}

We are now ready to prove the existence of the wave operator and a new representation for it.

\begin{prop}
The wave operator $W_-^\CD$ exists and is equal to the operator $W_-\big(H^{\CD},H_0\big)$ defined in \eqref{defW}.
In addition, the following equality holds in $L^2(\R;\C^2)$:
\begin{equation}\label{egalite}
\V\;\!N\;\!J^*\;\!\F_0\;\!(W_-^{\CD}-1)\;\!\F_0^*\;\!J\;\!N^*\;\!\V^*
= R(\D)\;\!T_0^{\CD}\Bigg(m\frac{\e^\X+1}{\e^\X-1}\Bigg).
\end{equation}
\end{prop}

\begin{proof}
Starting from the equality \eqref{ww2} and by taking Lemmas~\ref{lem5} and~\ref{lem6} into account, one easily deduces that the following equalities hold for any $\eta, \eta'$ in the dense subset $\SS$ of $L^2(\Sigma;\C^2)$:
\begin{align*}
&\langle \F_0(W_-^{\CD}-1)\F_0^*J\eta, J\eta'\rangle_{_{\!\!\Hrond}} \\
=&   \frac{1}{2}\int_\Sigma \lim_{\varepsilon \searrow 0}
\big\langle  T_\varepsilon^{\CD}(\lambda) B(\lambda)^{-1} \Gamma_2 \F_0^* \delta_\varepsilon(L_0-\lambda)J\eta,
i B(\lambda)^{-1}\Gamma_2 \F_0^*(L_0-\lambda+ i\varepsilon)^{-1} J\eta' \big\rangle_{\C^2}\;\!
\d \lambda \\
=&  \int_\Sigma
\big\langle  T_0^{\CD}(\lambda) N(\lambda)\eta(\lambda),
[\Theta_0 N\eta'](\lambda) \big\rangle_{\C^2}\;\!
\d \lambda.
\end{align*}
Now, the existence of the limit (inside the integral sign) for any $\lambda \in \Sigma$ and any $\eta,\eta'\in \SS$ already shows the existence of the stationary wave operator $W_-^{\CD}$, see \cite[Def.~2.7.2]{Y}. In addition, its equality with
$W_-\big(H^{\CD},H_0\big)$ follows from \cite[Thm.~5.2.4]{Y}. Finally,  relation \eqref{egalite} can be deduced by a conjugation with the unitary operators $N$ and $\V$.
\end{proof}

It only remains to link the multiplication operator $T_0^{\CD}(L)$ in $L^2(\Sigma;\C^2)$ with the scattering operator $S^\CD := S(H^\CD,H_0)$ introduced in \eqref{defS}. For that purpose, recall that since the operator $S^{\CD}$ commutes with $H_0$, the operator $\F_0\;\!S^{\CD}\;\!\F_0^*$ commutes with $L_0$ and, therefore,
corresponds to an operator of multiplication in $\Hrond$. For that reason, one usually writes $\F_0\;\!S^{\CD}\;\!\F_0^* = S^{\CD}(L_0)$, where $S^{\CD}(\lambda)$ is a unitary operator in $\Hrond(\lambda)$
for almost every $\lambda \in \Sigma$ which is called the scattering matrix at energy $\lambda$.

\begin{lemma}
For almost every $\lambda \in \Sigma$, the following equality holds:
\begin{equation}\label{eqTetS}
S^\CD(\lambda)=1+ N(\lambda)^*\;\!T_0^\CD(\lambda)\;\!N(\lambda).
\end{equation}
\end{lemma}

\begin{proof}
We proceed by using the intertwining relation and the invariance principe.
It is well known that
if $\alpha : \Sigma\to \R$ is smooth and has a positive derivative, then the scattering operator $S^{\CD}$
is the strong limit of the operators $\e^{it\alpha(H_0)} \;\!W_-^{\CD} \;\!\e^{-it\alpha(H_0)}$
as $t \to \infty$, see for example \cite[Sec.~2.6]{Y}.
Let us consider the function $\alpha : \Sigma \to \R$ defined by $\alpha(\lambda) := \ln\big(\frac{\lambda-m}{\lambda+m}\big)$. Clearly, this function is smooth on $\Sigma$ with a positive derivative, which gives
\[
s-\lim_{t\to \infty}\e^{it\alpha(H_0)} \;\!\big(W_-^{\CD}-1\big) \;\!\e^{-it\alpha(H_0)} = S^{\CD}-1.
\]
We also observe that due to \eqref{vv1} we have
\begin{equation}\label{rel2}
\e^{it\alpha(L)}\V^*=\V^* \e^{-it\X}.
\end{equation}
Now, by using successively relations \eqref{rel1}, \eqref{egalite}, \eqref{vv1} \eqref{defL} and the commutativity of $\e^{-i\alpha(L)}$ with $N$ and $T_0^{\CD}(L)$,  \eqref{rel2}, and the usual relation between the operators $\X$ and $\D$, one infers that
\begin{align*}
S^\CD(L_0)-1
= & \F_0 (S^{\CD}-1)\F_0^*\\
= & s-\lim_{t \to \infty}\F_0\;\!\e^{it\alpha(H_0)} \;\!\big(W_-^{\CD}-1\big) \;\!\e^{-it\alpha(H_0)} \F_0^* \\
= & s-\lim_{t \to \infty} \e^{it\alpha(L_0)}\;\!\F_0 \;\!\big(W_-^{\CD}-1\big) \;\! \F_0^* \;\!\e^{-it\alpha(L_0)} \\
= & s-\lim_{t \to \infty} \e^{it\alpha(L_0)}
\;\!J\;\!N^*\;\!\V^*
\;\! R(\D)\;\!T_0^{\CD}\Bigg(m\frac{\e^\X+1}{\e^\X-1}\Bigg)
\;\!\V\;\!N\;\!J^*
\;\!\e^{-it\alpha(L_0)} \\
= & s-\lim_{t \to \infty} \e^{it\alpha(L_0)}
\;\!J\;\!N^*\;\!\V^*
\;\! R(\D)\;\!\V\;\!T_0^{\CD}(L)
\;\!N\;\!J^*
\;\!\e^{-it\alpha(L_0)} \\
= & s-\lim_{t \to \infty}
J\;\!N^*\;\!\e^{it\alpha(L)}\;\!\V^*\;\!R(\D)\;\!\V\;\!\e^{-it\alpha(L)}\;\!T_0^{\CD}(L)\;\!N\;\!J^* \\
= & s-\lim_{t \to \infty}
J\;\!N^*\;\!\V^*\;\!\e^{-it\X}\;\!R(\D)\;\!\e^{it\X}\;\!\V\;\!T_0^{\CD}(L)\;\!N\;\!J^* \\
= & s-\lim_{t \to \infty}
J\;\!N^*\;\!\V^*\;\!R(\D+t)\;\!\V\;\!T_0^{\CD}(L)\;\!N\;\!J^*.
\end{align*}
Finally, since $s-\lim_{t \to \infty}R(\D+t)=1$, it directly follows from the relation \eqref{defL} between $L$ and $L_0$ that
$$
S^\CD(L_0)-1 = \ J\;\!N^*\;\!T_0^{\CD}(L)\;\!N\;\!J^*
= N^*(L_0)\;\!T_0^{\CD}(L_0)\;\!N(L_0).
$$
The statement is then a consequence of the pointwise identification of these two multiplication operators.
\end{proof}

\begin{rem}
By taking the relation $ W_+^{\CD} = W_-^{\CD} (S^{\CD})^*$ into account, an explicit expression for $W_+^{\CD}$, similar to the one obtained in \eqref{egalite} for $W_-^{\CD}$, could also be derived.
\end{rem}

\begin{rem}\label{rem37}
As a consequence of Lemma \ref{lem-unt} and of the explicit formula \eqref{defR}, the maps
\[
x\mapsto R(x), \quad x\mapsto T^\CD_0\Big(m\dfrac{\e^x+1}{\e^x-1}\Big)
\]
are continuous on the whole real line and admit limits at $\pm\infty$.
This is an essential feature of these functions, and it plays an essential role in the subsequent algebraic construction.
\end{rem}

\section{Topological results}\label{sectop}

In this section, we briefly deduce the main corollary of the explicit formula \eqref{egalite}, and refer to \cite{KPR} and \cite{RT1} for a thorough description of the underlying algebraic framework.

Let us start by defining the following map: For $x,y \in \R$ one sets
\begin{equation}\label{defG}
\Gamma^{\CD}(x,y) = 1 + R(y) \;\!T_0^{\CD}\Big(m\frac{\e^x+1}{\e^x-1}\Big),
\end{equation}
where $R(\cdot)$ has been introduced in \eqref{defR} and $T_0^{\CD}(\cdot)$ has been computed explicitly in \eqref{defT0}.
It follows from Remark~\ref{rem37}
that $\Gamma^{\CD}$ can be continuously extended to a function on $\blacksquare:=[-\infty,+\infty]\times[-\infty,+\infty]$.
More precisely, one can set
$\Gamma^{\CD} \in C\big(\blacksquare; M_2(\C)\big)$ with $\Gamma^{\CD}(x,y)$ provided by \eqref{defG}.
The asymptotic values of this function can then be easily computed, namely
\begin{align*}
\Gamma_1^{\CD}(y):= &\Gamma^{\CD}(-\infty,y) =  1+ R(y) \;\!T_0^{\CD}(m)\\
\Gamma_2^{\CD}(x):= &\Gamma^{\CD}(x,+\infty) =  1 + T_0^{\CD}\Big(m\frac{\e^x+1}{\e^x-1}\Big)\\
\Gamma_3^{\CD}(y):= &\Gamma^{\CD}(+\infty,y) =  1 + R(y) \;\!T_0^{\CD}(-m)\\
\Gamma_4^{\CD}(x):= &\Gamma^{\CD}(x,-\infty) =  1.
\end{align*}
It is certainly worth emphasizing that $\Gamma_1^{\CD}$ and $\Gamma_3^{\CD}$ are related to the behavior of $T_0^{\CD}$ at the thresholds values $\pm m$, while $\Gamma_2^{\CD}$ is related to the scattering operator through the relation \eqref{eqTetS}. Note also that the precise value of $T_0^{\CD}(\pm m) \in M_2(\C)$ could be explicitly computed in terms of $C$ and $D$, but that this is not our concern here (a similar computation has been performed for example in \cite[Prop.~14]{PR} for the Aharonov-Bohm operator).

Let us now observe that the boundary $\square$ of $\blacksquare$ consists in the union of four parts $B_1\cup B_2 \cup B_3 \cup B_4$, with $B_1=\{-\infty\}\times [-\infty,+\infty]$, $B_2=[-\infty,+\infty]\times\{+\infty\}$, $B_3=\{+\infty\}\times [-\infty,+\infty]$ and $B_4=[-\infty,+\infty]\times\{-\infty\}$. Therefore, one can define the function
$$
\Gamma^{\CD}_{\!\!{}^\square}: \square \to M_2(\C)
$$
with $\Gamma^{\CD}_{\!\!{}^\square} \big|_{B_j}=\Gamma^{\CD}_j$. By construction, the function $\Gamma^{\CD}_{\!\!{}^\square}$ is continuous and takes values in $U(2)$, the subset of unitary matrices in $M_2(\C)$, or more precisely
$\Gamma^{\CD}_{\!\!{}^\square} \in C\big(\square;U(2)\big)$.
Note that the property $\Gamma^{\CD}_{\!\!{}^\square}(\theta)\in U(2)$ for any $\theta\in \square$ can be checked explicitly, but also directly follows from the unitarity of the image of the wave operators in the Calkin algebra.

One of the main result of the $C^*$-algebraic framework which has been developed for example in \cite{KPR} and \cite{RT1} is to relate the function $\Gamma^{\CD}_{\!\!{}^\square}$ to the number of bound states of $H^{\CD}$. More precisely, let us define the winding number $\wind[\Gamma^{\CD}_{\!\!{}^\square}]$ of the map
$$
\square \ni \theta \mapsto \det\big[\Gamma^{\CD}_{\!\!{}^\square}(\theta)\big] \in \T
$$
with orientation of $\square$ chosen clockwise. Here, $\T$ denotes the set of complex numbers of modulus $1$, and $\det$ denotes the usual determinant on $M_2(\C)$. Then, the following topological version of Levinson's theorem holds:

\begin{theorem}\label{thmtop}
For any admissible pair $(C,D)$ one has
\begin{equation}\label{levtop}
\wind[\Gamma^{\CD}_{\!\!{}^\square}] = - \#\spec_p (H^{\CD}).
\end{equation}
\end{theorem}

Once in the suitable $C^*$-algebraic framework, the proof of this statement is quite standard. We refer to \cite{KR5,RT1} for the construction of the framework and for the related proofs.

As a final remark, let us comment of the contribution of each term $\Gamma^{\CD}_j$ for $j \in \{1,2,3,4\}$ in the l.h.s.~of \eqref{levtop}. Clearly, the contribution of $\Gamma^{\CD}_4$ is trivial, while the contributions of $\Gamma^{\CD}_1$ and $\Gamma^{\CD}_3$ are directly related to the threshold effects at $\pm m$. These effects do depend on the choice of the pair $(C,D)$. For the contribution of $\Gamma^{\CD}_2$, let us observe that
$$
\det\big[\Gamma^{\CD}_2(\theta)\big] =  \det\Big[1 + T_0^{\CD}\Big(m\frac{\e^\theta+1}{\e^\theta-1}\Big)\Big]
= \det\Big[S^{\CD}\Big(m\frac{\e^\theta+1}{\e^\theta-1}\Big)
\Big]
$$
where relation \eqref{eqTetS} has been taken into account.
Thus, when $\theta$ varies from $-\infty$ to $+\infty$, one easily observes that the contribution due to $\Gamma^{\CD}_2$ is provided by two distinct contributions, the one coming from $\det\big[S^{\CD}(\lambda)\big]$ as $\lambda$ runs from $m$ to $+\infty$, and the one coming from $\det\big[S^{\CD}(\lambda)\big]$ as $\lambda$ goes from $-m$ to $-\infty$.
Note that the difference of relative orientation for the two contributions was already noticed in the literature, see for example \cite{C89,K90}. Note also that even if $S^{\CD}(-\infty)\neq S^{\CD}(+\infty)$, the equality
$\det\big[S^{\CD}(-\infty)\big]=\det\big[S^{\CD}(+\infty)\big]$ holds, as shown in the above computations.

\section*{Appendix}

\begin{proof}[Proof of Lemma~\ref{lem-unt}]
We first observe that for any $\lambda\in\Sigma$ we can write
\begin{equation}\label{eq-mle}
M(\lambda+i\varepsilon)=i B(\lambda)^2+  K(\lambda,\varepsilon),
\end{equation}
with $K(\lambda,\varepsilon)\to 0$ as $\varepsilon\searrow0$ locally uniformly in $\lambda \in \Sigma$. The scheme of the following argument is similar to the one already used in the proof of Proposition \ref{propspectral}.

a) Consider first the case $\det(D)\ne 0$, and set $\Lambda:=D^{-1} C$. Then, $\Lambda^*=\Lambda$ and
$T_\varepsilon^{\CD}(\lambda)=-2iB(\lambda) \big(M(\lambda+i\varepsilon)-\Lambda)^{-1}\!B(\lambda)$.
As $B(\lambda)>0$ for any $\lambda\in\Sigma$, the continuity statement follows
from the representation \eqref{eq-mle} and from the continuity of the maps $\Sigma\ni\lambda\to B(\lambda) \in M_2(\C)$ and $\Sigma\times[0,+\infty)\ni(\lambda,\varepsilon)\mapsto M(\lambda+i\varepsilon)\in M_2(\C)$.
In addition, we get after some elementary algebra that
\begin{align*}
& \big(1+T_0^{\CD}(\lambda)\big)\big(1+T_0^{\CD}(\lambda)\big)^*-1\\
= &4B(\lambda)\big( M(\lambda+i0)-\Lambda\big)^{-1}\Big(B(\lambda)^2-\dfrac{M(\lambda+i0)-M(\lambda+i0)^*}{2i}\Big)
\big( M(\lambda+i0)^*-\Lambda\big)^{-1}B(\lambda)\\
=&0,
\end{align*}
which shows the unitarity of $1+T_0^{\CD}(\lambda)$. The existence of the limits can be checked directly. In particular,
the equality $T_0^{\CD}(-\infty)=T_0^{\CD}(+\infty)$ follows from $M(-\infty+i0)=M(+\infty+i0)=\frac{i}{2}$.

b) We now consider the case $\dim[\ker (D)]=1$ and proceed as in \cite[Sec.~3]{PR}.
Let $I:\C \to \C^2$ be the identification of $\C$ with $\ker(D)^\bot$, and let $P$ denote its adjoint, {\it i.e.}~$P:\C^2 \to \C$ is the composition of the orthogonal projection onto $\ker(D)^\bot$ together with the identification of $I \C$ with $\C$.
Then, as shown in \cite[Eq. (12)]{PR}, one has $(DM(z)-C)^{-1}D=I\big(m(z)-\ell\big)^{-1}P$, where $\ell := (DI)^{-1}CI \in \R$
and $m(z)=PM(z)I\in \C$. One easily checks that $m(\lambda+i\varepsilon)=i\beta(\lambda)+k(\lambda,\varepsilon)$ with
$k(\lambda,\varepsilon)\to 0$ as $\varepsilon\searrow0$ locally uniformly in $\lambda \in \Sigma$, with $\beta(\lambda)=PB(\lambda)^2I\in \R$,
and $\beta(\lambda)>0$ for $\lambda\in\Sigma$. The continuity statement now follows
from the continuity of the maps $\Sigma\ni\lambda\to \beta(\lambda)\in \R$ and $\Sigma\times[0,+\infty)\ni(\lambda,\varepsilon)\mapsto m(\lambda+i\varepsilon)\in \C$.
By using $P^*=I$, we then obtain
\begin{align*}
&\big(1+T_0^{\CD}(\lambda)\big)\big(1+T_0^{\CD}(\lambda)\big)^*- 1\\
=&\Big(1-2iB(\lambda)I \big(m(\lambda+i0)-\ell\big)^{-1}PB(\lambda)\Big)\Big(1-2iB(\lambda)I \big(m(\lambda+i0)-\ell\big)^{-1}PB(\lambda)\Big)^*-1
\\
=&4\big(m(\lambda+i0)-\ell\big)^{-1}\big( m(\lambda+i0)^*-\ell\big)^{-1}
B(\lambda)I\Big(PB(\lambda)^2I-\dfrac{m(\lambda+i0)-m(\lambda+i0)^*}{2i}\Big)PB(\lambda) \\
=&0,
\end{align*}
which shows the unitarity. The existence of the limits at the boundary can be checked explicitly.

c) The remaining case $D=0$ is trivial.

\end{proof}


\begin{thebibliography}{00}

\bibitem{alonso} V. Alonso, S. De Vincenzo:
\emph{Delta-type Dirac point interactions and their nonrelativistic limits},
Intern. J. Theor. Phys. {\bf 39} (2000), 1483--1498.

\bibitem{BMN}
J. Behrndt, M. Malamud, H. Neidhardt:
\emph{Scattering matrices and Weyl functions},
Proc. Lond. Math. Soc. (3) {\bf 97} (2008), 568--598.

\bibitem{BS}
J. Bellissard, H. Schulz-Baldes:
\emph{Scattering theory for lattice operators in dimension $d\geq 3$}, Rev. Math. Phys. {\bf 24} (2012), no. 8, 1250020.

\bibitem{dabr} S. Benvegn\`u, L. D\c{a}browski:
\emph{Relativistic point interaction},
Lett. Math. Phys. {\bf 30} (1994), 159--167.

\bibitem{BGP}
J. Br\"uning, V. Geyler, K. Pankrashkin:
\emph{Spectra of self-adjoint extensions and applications to solvable Schr\"odinger operators},
Rev. Math. Phys. {\bf 20} (2008), 1--70.

\bibitem{cmp}
R. Carlone, M. Malamud, A. Posilicano:
\emph{On the spectral theory of Gesztesy-–\v Seba realizations of 1-D Dirac operators with point interactions on a discrete set},
J. Differential Equations {\bf 254} (2013), 3835--3902.

\bibitem{C89}
D. Clemence:
\emph{Low-energy scattering and Levinson's theorem for a one-dimensional Dirac equation},
Inverse Problems {\bf 5} (1989), no. 3, 269--286.

\bibitem{DM}
V.A. Derkach, M.M. Malamud:
\emph{Generalized resolvents and the boundary value problems for Hermitian operators with gaps},
J. Funct. Anal. {\bf 95} (1991), 1--95.

\bibitem{gg}
V.~I.~Gorbachuk, M.~L.~Gorbachuk:
\emph{Boundary value problems for operator differential equations},
Mathematics and its Applications, Soviet Series,
Vol.~48, Kluwer Acad. Publ., Dordrecht etc., 1991.

\bibitem{IR}
H. Isozaki, S. Richard:
\emph{On the wave operators for the Friedrichs-Faddeev model},
Ann. Henri Poincar\'e {\bf 13} (2012), 1469--1482.

\bibitem{Ito}
H.T. Ito:
\emph{High-energy behavior of the scattering amplitude for a Dirac operator},
Publ. Res. Inst. Math. Sci. {\bf 31} (1995), 1107--1133.

\bibitem{Jak}
V. Jak\v si\'c:
\emph{Topics in spectral theory}, In: S.~Attal, A.~Joye, C.-A.~Pillet (Eds.),
\emph{Open Quantum Systems I. Recent Developments}, pp. 235--312,
Lecture Notes Math., Vol. 1880, Springer, Berlin, 2006.

\bibitem{Jef}
A. Jeffrey:
\emph{Handbook of mathematical formulas and integrals},
Academic Press, Inc., San Diego, CA, 1995.

\bibitem{KPR}
J. Kellendonk, K. Pankrashkin, S. Richard:
\emph{Levinson's theorem and higher degree traces for Aharonov-Bohm operators},
J. Math. Phys. {\bf 52} (2011), 052102.

\bibitem{KR1}
J. Kellendonk, S. Richard:
\emph{Levinson's theorem for Schr\"odinger operators with point interaction: a topological approach},
J. Phys. A {\bf 39} (2006), no. 46, 14397--14403.

\bibitem{KR5}
J. Kellendonk, S. Richard:
\emph{On the structure of the wave operators in one dimensional potential scattering}, Mathematical Physics Electronic Journal {\bf 14} (2008), 1--21.

\bibitem{KR6}
J. Kellendonk, S. Richard:
\emph{On the wave operators and Levinson's theorem for potential scattering in $\R^3$},
Asian-European Journal of Mathematics {\bf 5} (2012), 1250004-1--1250004-22.

\bibitem{K90}
M. Klaus:
\emph{On the Levinson theorem for Dirac operators},
J. Math. Phys. {\bf 31} (1990), no. 1, 182--190.

\bibitem{L07}
D.-H. Lin:
\emph{Friedel sum rule, Levinson theorem, and the Atiyah-Singer index},
Phys. Rev. A {\bf 75} (2007), 032115.

\bibitem{L99}
Q.-G. Lin:
\emph{Levinson theorem for Dirac particles in one dimension},
Eur. Phys. J. D At. Mol. Opt. Phys. {\bf 7} (1999), no. 4, 515--524.

\bibitem{MN85}
Z.-Q. Ma, G.-J. Ni:
\emph{Levinson theorem for Dirac particles},
Phys. Rev. D {\bf 31} (1985), 1482-–1488.

\bibitem{PR}
K. Pankrashkin, S. Richard:
\emph{Spectral and scattering theory for the Aharonov-Bohm operators},
Rev. Math. Phys. {\bf 23} (2011), 53--81.

\bibitem{SB}
H. Schulz-Baldes:
\emph{The density of surface states as the total time delay},
Preprint arXiv:1305.2187.

\bibitem{schm}
K. Schm\"udgen:
\emph{Unbounded self-adjoint operators on Hilbert space},
Graduate Texts in Mathematics, Vol.~265, Springer, 2012.

\bibitem{RT1}
S. Richard, R. Tiedra de Aldecoa:
\emph{New formulae for the wave operators for a rank one interaction},
Integral Equations and Operator Theory {\bf 66} (2010), 283--292.

\bibitem{RT2}
S. Richard, R. Tiedra de Aldecoa:
\emph{New expressions for the wave operators of Schr\"odinger operators in $\R^3$},
Lett. Math. Phys. {\bf 103} (2013), 1207--1221.

\bibitem{RT3}
S. Richard, R. Tiedra de Aldecoa:
\emph{Explicit formulas for the Schr\"odinger wave operators in $\R^2$},
C. R. Acad. Sci. Paris, Ser. I. {\bf 351} (2013), 209--214.

\bibitem{Y}
D.R. Yafaev:
\emph{Mathematical scattering theory. General theory},
Translations of Mathematical Monographs,
Vol.~105, AMS, Providence, RI, 1992.

\end{thebibliography}
\end{document}